\documentclass{article}

\usepackage{graphicx, amsfonts, amsmath, amssymb, amsthm, xcolor, cleveref, mathtools}

\usepackage{tikz}
\usetikzlibrary{matrix}
\usetikzlibrary{decorations.pathreplacing, positioning}
\usepackage{graphicx}
\usepackage{authblk}
\usepackage{fullpage, braket}
\usepackage[ruled,vlined,linesnumbered]{algorithm2e}
\usepackage{algpseudocode, quantikz}
\usepackage{url}
\usepackage{authblk}
\usepackage{soul}
\usepackage{bm,bbm}

\newtheorem{theorem}{Theorem}

\newtheorem{remark}{Remark}
\newtheorem{definition}{Definition}

\newtheorem{example}{Example}
\newtheorem{lemma}{Lemma}

\newif\ifshow
\showfalse

\title{Commuting Embeddings for Parallel Strategies in Non-local Games}

\author[1,2,$\S$]{Sarah Chehade}
\author[3,$\dagger$]{Andrea Delgado}
\author[2,$\ddagger$]{Elaine Wong}

\affil[1]{UTC Quantum Center, University of Tennessee, Chattanooga, TN 37403, USA}
\affil[2]{Computer Science $\&$ Mathematics Division, Oak Ridge National Laboratory, Oak Ridge, TN 37831, USA\thanks{This manuscript has been authored by UT-Battelle, LLC under Contract No. DE-AC05-00OR22725 with the U.S. Department of Energy. The United States Government retains and the publisher, by accepting the article for publication, acknowledges that the United States Government retains a non-exclusive, paid-up, irrevocable, world-wide license to publish or reproduce the published form of this manuscript, or allow others to do so, for United States Government purposes. The Department of Energy will provide public access to these results of federally sponsored research in accordance with the DOE Public Access Plan (http://energy.gov/downloads/doe-public-access-plan).} }
\affil[3]{Physics Division, Oak Ridge National Laboratory, Oak Ridge, TN 37831, USA}
\affil{
  $\S$ sarah-chehade@utc.edu \quad
  $\dagger$ delgadoa@ornl.gov \quad
  $\ddagger$ wongey@ornl.gov
}

\date{\today}

\begin{document}

\maketitle

\begin{abstract}
    Non-local games (NLGs) provide a versatile framework for probing quantum correlations and for benchmarking the power of entanglement. In finite dimensions, the standard method for playing several games in parallel requires a tensor product of the local Hilbert spaces, which scales additively in the number of qubits. In this work, we show that this additive cost can be reduced by exploiting algebraic embeddings. We introduce two forms of compressions. First, when a referee selects one game from a finite collection of games at random, the game quantum strategy can be implemented using a maximally entangled state of dimension equal to the largest individual game, thereby eliminating the need for repeated state preparations. Second, we establish conditions under which several games can be played simultaneously in parallel on fewer qubits than the tensor product baseline. These conditions are expressed in terms of commuting embeddings of the game algebras. Moreover, we provide a constructive framework for building such embeddings. Using tools from Lie theory, we show that aligning the various game algebras into a common Cartan decomposition enables such a qubit reduction. Beyond the theoretical contribution, our framework casts NLGs as algebraic primitives for distributed and resource constrained quantum computations and suggested NLGs as a comparable device independent dimension witness.     
\end{abstract}

\section{Introduction}\label{sec: intro}

Non-local games (NLGs) provide an operational framework for probing quantum correlations. In the standard setting, two or more non-communicating players receive classical questions from a referee and must respond with classical answers. The winning condition depends jointly on their inputs and outputs. A strategy is \textit{classical} if players rely only on shared randomness, and \textit{quantum} if they share entanglement and perform local measurements. Certain games admit strictly higher winning probabilities under quantum strategies, certifying correlations beyond classical resources.

Being able to witness non-classical correlations between devices makes NLGs central to \textit{device-independent} quantum information. Since the game referee only sees classical input-output statistics, violations of classical bounds can certify the presence of entanglement without assumptions about internal states or measurements. But even beyond entanglement certification, NLGs enable device-independent generation of randomness~\cite{pan2021oblivious}, self-testing of quantum states and measurements~\cite{vsupic2020self}, and verification of non-local correlations.

In this paper, we study how \textit{multiple games with perfect strategies} can be realized simultaneously within a single Hilbert space, without requiring the full tensor product of the underlying spaces. Formally, if each game $\mathcal{G}_i$ admits a perfect strategy on $2^{n_i}$-dimensional local spaces, naive parallel composition requires dimension $2^{\sum_i n_i}$. We show that, under algebraic conditions, the same strategies can be implemented in a Hilbert space of dimension $2^n$ for some $n < \sum_i n_i$, by embedding the operator algebras into commuting subalgebras. Previous work on “gluing” NLGs~\cite{cui2020generalization, manvcinska2021glued} merges rule sets to characterize states or operators, often for applications to self-testing. Our approach presents a different focus with two key benefits: it separates the additive dimension bound from the minimal algebraic bound, and offers a constructive view of NLGs as algebraic primitives for parallel tasks. Families of strategies can coexist in reduced dimension whenever their algebras commute, tying NLGs to themes in operator algebras and representation theory.

To frame the broader work of NLGs in a proper theoretical context, its landscape is derived from Tsirelson-type bounds~\cite{todorov2024quantum}, rigidity results~\cite{mousavi2022nonlocal}, self-testing theorems~\cite{miller2012optimal}, with deep connections to quantum complexity theory~\cite{Ji2021}. Canonical games such as CHSH~\cite{clauser1969proposed} and Mermin–Peres~\cite{arkhipov2012extending} have illuminated the limits of quantum strategies and their computational hardness. While CHSH has been realized across platforms, scaling to richer games is challenging—particularly implementing multiple device-independent tests on the same hardware, or playing families of games in parallel with explicit resource accounting. Recent work~\cite{furches2025application} shows NLGs can also serve as practical benchmarks for near-term quantum devices, signaling a shift toward operational deployment. In contrast, algebraic features of a game can can be directly used to design quantum algorithms to play NLGs~\cite{chehade2025game}. This highlights algorithmic potential for NLGs to explicitly construct game strategies, which are needed for scalable hardware implementations.

Yet, a unifying methodology connecting the algebraic structure of NLGs to hardware resource constraints is missing. One promising direction explored in this paper casts NLGs as algorithmic primitives. We formalize a qubit-reduction method for playing multiple NLGs in parallel by embedding each perfect strategy into commuting blocks within a shared Hilbert space. Then, instead of focusing on self-testing applications, our goal is qubit compression motivated by global consistency constraints, with applications to distributed networks and resource-efficient computation. In other words, \textbf{how do local constraints combine to form a global consistency condition} for applications such as a power grid or communications network setup?

    The remainder of the paper is organized as follows:
    \Cref{sec: prelims} defines general non-local games and fixes the notations used throughout. \Cref{sec: warmup} formally explains the steps to play parallel games together without any reductions, via tensor products. \Cref{sec: dependent} explains the results of potential reductions while maintaining perfect quantum strategies for a finite collection of games. Finally, \Cref{sec: Cartan} discusses the role of the Lie algebra, for which we use to search for the optimal measurement operators. A pseudo-algorithm is presented here to summarize how to leverage the different results for game compression.

\section{Preliminaries}\label{sec: prelims}
    In this section, we fix our notations and assumptions. Throughout this manuscript, we restrict our attention to finite-dimensional Hilbert spaces. We focus on two-player non-local games for simplicity, however our results may extend to any finite number of players.
    
    A non-local game $\mathcal{G}$ is specified by input pairs $(x,y)\in \mathcal{I}_A\times \mathcal{I}_B$ and output pairs $(a,b)\in \mathcal{O}_A\times \mathcal{O}_B$ along with a probability distribution function $\mu$ over inputs and a rule function $\lambda: \mathcal{I}_A\times \mathcal{I}_B\times \mathcal{O}_A\times \mathcal{O}_B\to \{0,1\}$. The referee samples $(x,y)\sim \mu$ and sends $x$ to player Alice and $y$ to player Bob. The players reply with the outputs $(a,b)$ based on a predetermined strategy that we will define below. The players win if and only if $\lambda(x,y,a,b)=1.$
    When we write $\vec{x} = (x_1,\dots,x_K)$ and $\vec{y} = (y_1,\dots,y_K)$,  we mean $K$ simultaneous questions posed to Alice and Bob, respectively,  corresponding to playing $K$ parallel instances of the base game.  The outputs are then also denoted $\vec{a}=(a_1,\dots,a_K)$ and $\vec{b}=(b_1,\dots,b_K)$. A game is called \textit{synchronous} provided that $\mathcal{I}_A=\mathcal{I}_B$ and $\mathcal{O}_A=\mathcal{O}_B$. 
    
    A \textit{classical strategy} consists of deterministic response functions
    \begin{equation}
        f:\mathcal{I}_A\to\mathcal{O}_A \, \quad\text{and}\quad \, g:\mathcal{I}_B\to\mathcal{O}_B.  
    \end{equation}
    A \textit{quantum strategy} consists of finite-dimensional Hilbert spaces per player $\mathcal{H}_A$ and $\mathcal{H}_B$, a shared entangled state $\ket{\psi}\in\mathcal{H}_A\otimes \mathcal{H}_B$, and positive operator-valued measures (POVMs)  $\{M_{x,a}\}_a$ and $\{N_{y,b}\}_b$ on $\mathcal{H}_A$ and $\mathcal{H}_B$, respectively. The induced answer distribution from a quantum strategy is then 
    \begin{equation}
        p_q(a,b|x,y) \;=\; \bra{\psi}\, M_{x,a}\otimes N_{y,b}\,\ket{\psi}.
    \end{equation}
        The \textit{classical game value} of a game $\mathcal{G}$ is defined as
        \begin{equation}
            \mathrm{Val}_c(\mathcal{G}) = \sup_{f,g}   \sum\limits_{x,y}\mu(x,y)\sum\limits_{a,b} \lambda(x,y,a,b)\, \mathrm{Pr}[f(x)=a, g(y)=b],
        \end{equation}
        where the supremum is taken over all classical strategies.

        The \textit{quantum game value} of a game $\mathcal{G}$ is defined as
        \begin{equation}
             \mathrm{Val}_q(\mathcal{G})= \sup_{p_{\mathrm{q}}}  \sum_{x,y}\mu(x,y)\sum_{a,b} \lambda(x,y,a,b)\, p_{\mathrm{q}}(a,b|x,y),
        \end{equation}
        where the supremum is taken over all quantum strategies. 
    A strategy is called \textit{perfect}, if it achieves a game value of 1. 
    Let $\mathcal{G}$ be a game with a perfect strategy. The \textit{game algebra} for Alice is a unital $C^*-$algebra generated by her measurement operators   $\{M_{x,a}\}_a$ subject to the relations imposed by the perfect strategy constraints. Similarly, Bob has a game algebra. These algebras encode all operator identities that must hold in any perfect strategy for the game. 

    In this paper, we focus on various ways to embed games into larger structures while maintaining game strategies but reducing the required amount of qubits for future applications. This definition is formalized as follows:
        Let $\bm{\mathcal{G}}=\{\mathcal{G}_i\}_{i=1}^K$ be a collection on non-local games with perfect strategies on $n_i$ qubits per player, generating algebras $\mathcal{A}_i$ and $\mathcal{B}_i$. We say that $\bm{\mathcal{G}}$ compresses into a single game on $n$ qubits if there exists
        \begin{itemize}
            \item embeddings $\iota_i^A:\mathcal{A}_i\to \mathcal{B}(\mathbb{C}^{2^n})$ and $\iota_i^B:\mathcal{B}_i\to \mathcal{B}(\mathbb{C}^{2^n})$, whose ranges commute across $i$. i.e., $[X,Y] = 0$, for all   $X \in \iota_i^A(\mathcal{A}_i) \cup \iota_i^B(\mathcal{B}_i)$ and   $Y \in \iota_j^A(\mathcal{A}_j) \cup \iota_j^B(\mathcal{B}_j)$, whenever $i \neq j$ and
            \item a single bipartite state $\ket{\Psi}\in\mathbb{C}^{2^n}\otimes \mathbb{C}^{2^n}$,
        \end{itemize}
        such that the embedded measurements reproduce the winning strategies of each game $\mathcal{G}_i$.

\section{Baseline Games}\label{sec: warmup}
    In this section, we observe the baseline independent case before considering any qubit reductions across games. This case serves two purposes. First, it establishes how perfect strategies theoretically tensor without cross-game inference, which may be used as a \textit{hardware benchmark} for quantum devices. Second, this notion may be used as a \textit{resource benchmark} by counting qubit costs. Any subsequent reduction result must therefore beat this additive baseline while preserving the perfect strategy.

    In addition, we consider the `simplest' compression, where the referee samples uniformly from a finite number of NLGs. This result shows how a single maximally entangled state of local dimension equal to the largest game suffices. 

\begin{theorem}\label{thm: indep}
     Let $\{\mathcal{G}_i\}_{i=1}^K$ be $K$ non-local games, each admitting a perfect quantum strategy on a fixed number of qubits $n_i$ for each player. 
    Let $\mathcal{A}_i$ be the $C^*$-algebra generated by the measurement operators of $\mathcal{G}_i$ for Alice, and $\mathcal{B}_i$ the corresponding algebra for Bob. 
    Suppose that for each game $\mathcal{G}_i$ there exists a shared entangled state 
    \begin{equation}
        \ket{\psi_i} \in \mathbb{C}^{2^{n_i}} \otimes \mathbb{C}^{2^{n_i}}
    \end{equation}
    and measurement operators $\{M_{x_i,a_i}\}_a \subset \mathcal{A}_i$ and $\{N_{y_i,b_i}\}_b \subset \mathcal{B}_i$ 
    achieving the winning conditions of $\mathcal{G}_i$ with probability $1$. 
    Then there exists a joint state
    \begin{equation}
        \ket{\Psi} \;=\; \bigotimes_{i=1}^K \ket{\psi_i} \;\in\; \mathbb{C}^{2^N} \otimes \mathbb{C}^{2^N}, \text{ for } N := \sum_{i=1}^K n_i,
    \end{equation}
    and measurement operators 
    \begin{equation}
         M_{\vec{x},\vec{a}} = \bigotimes_{i=1}^K M_{x_i,a_i}, 
        \quad N_{\vec{y},\vec{b}} = \bigotimes_{i=1}^K N_{y_i,b_i},
    \end{equation}
    such that they define a perfect quantum strategy. i.e., the players can win all $K$ games played in parallel with probability $1$. 
\end{theorem}

\begin{proof}
    Let $i\in\{1,\dots,K\}$. By assumption, $\mathcal{G}_i$ has a perfect strategy. That is, for every question pair $(x_i,y_i)$, the outcome distribution produced by the POVMs $\{M_{x_i,a_i}\}_{a_i}$ and $\{N_{y_i,b_i}\}_{b_i}$ on $\ket{\psi_i}$ is
    \begin{equation}\label{eq: win per game}
        \sum_{\substack{a_i,b_i:\\ \lambda_i(x_i,y_i,a_i,b_i)=1 }} \bra{\psi_i}\,M_{x_i,a_i}\otimes N_{y_i,b_i}\,\ket{\psi_i}\;=\;1.
    \end{equation}
    Define the global Hilbert space as
    \begin{equation}     
        \mathcal{H}_A\;=\;\bigotimes_{i=1}^K \mathbb{C}^{2^{n_i}}  \quad\text{and}\quad \mathcal{H}_B\;=\;\bigotimes_{i=1}^K \mathbb{C}^{2^{n_i}},
    \end{equation}
    the global state $\ket{\Psi}=\bigotimes\limits_{i=1}^K \ket{\psi_i}$, and for any question tuple $(\vec{x},\vec{y})=((x_1,\dots,x_K),(y_1,\dots,y_K))$, define product POVMs
    \begin{equation*}
        \{M_{\vec{x},\vec{a}}\}_{\vec{a}}=\bigotimes_{i=1}^K \{M_{x_i,a_i}\}_{a_i} \quad\text{and}\quad \{N_{\vec{y},\vec{b}}\}_{\vec{b}}=\bigotimes_{i=1}^K \{N_{y_i,b_i}\}_{b_i}.
    \end{equation*}
    These are valid POVMs since the tensor product of POVMs is a POVM.
    The success probability of the product strategy on questions $(\vec{x},\vec{y})$ is 
    \begin{equation*}
        \sum_{\vec{a},\vec{b}}\bra{\Psi}\,M_{\vec{x},\vec{a}}\otimes N_{\vec{y},\vec{b}}\,\ket{\Psi}.
    \end{equation*}
    Using $\ket{\Psi}=\bigotimes\limits_i \ket{\psi_i}$ and $M_{\vec{x},\vec{a}}\otimes N_{\vec{y},\vec{b}} =\bigotimes\limits_i \big(M_{x_i,a_i}\otimes N_{y_i,b_i}\big)$, we factorize 
    \begin{equation*}
    \bra{\Psi}\,M_{\vec{x},\vec{a}}\otimes N_{\vec{y},\vec{b}}\,\ket{\Psi} =\prod_{i=1}^K \bra{\psi_i}\,M_{x_i,a_i}\otimes N_{y_i,b_i}\,\ket{\psi_i}
    \end{equation*}
    giving 
    \begin{equation*}
        \sum_{\vec{a},\vec{b}}\bra{\Psi}\,M_{\vec{x},\vec{a}}\otimes N_{\vec{y},\vec{b}}\,\ket{\Psi} =\prod_{i=1}^K \left(  \sum_{\substack{a_i,b_i:\\ \lambda_i(x_i,y_i,a_i,b_i)=1 }} \bra{\psi_i}\,M_{x_i,a_i}\otimes N_{y_i,b_i}\,\ket{\psi_i} \right).
    \end{equation*}
    By \Cref{eq: win per game}, each factor equals $1$, so the product also equals $1$. Thus, for every question pair $(\vec{x},\vec{y})$, the product strategy wins with probability $1$. That is, it is a perfect strategy for the parallel composition. 
    
    Finally, note that the total local dimension is $\dim \mathcal{H}_A=\dim \mathcal{H}_B=2^N$ with $N=\sum_i n_i$. The measurements act on disjoint tensor products and thus commutes across games, ensuring no cross-game interference.
\end{proof}

\begin{example}
    Take 4 games of Mermin's magic square game~\cite{Mermin1990b} and denote this set of games as $\bm{\mathcal{G}}^{\text{MSG}}$. See Table~\ref{tab:fourmsgs} for an example, where the elements in each square correspond to the POVMs used to compute the quantum strategy. Each MSG $\mathcal{G}_i$ requires 2 Bell states $\ket{\psi}_i = \ket{\phi^+}_{AB}\otimes\ket{\phi^+}_{AB}$. In addition, each game has a perfect strategy using $\bm{\mathcal{G}}^{\text{MSG}}$. By taking a tensor product across games, we use $\ket{\Psi}$ to be a tensor product of 8 Bell states (2 per game) and the measurement operators $M_{\vec{x},\vec{a}}$ and $N_{\vec{y},\vec{b}}$ are constructed by taking tensor products of all the measurement operators in each game. Figure 1 depicts how these games are aligned.  

    \begin{table}[ht]
    \centering
    \label{tab:fourmsgs}

    \setlength{\tabcolsep}{10pt}
    \renewcommand{\arraystretch}{1.2}
    
    % Outer grid
    \begin{tabular}{@{}c c@{}}
      \\ 
      % Top-left MSG
      \begin{tabular}{|c|c|c|}
        \hline 
        $Y\otimes X$ & $X \otimes Y$ & $Z \otimes Z$ \\ \hline
        $Y\otimes \mathbbm{1}$ & $\mathbbm{1}\otimes Y$ & $Y \otimes Y$ \\ \hline
        $\mathbbm{1}\otimes X$ & $X \otimes \mathbbm{1}$ & $X \otimes X$ \\
        \hline
      \end{tabular}
      &
      % Top-right MSG
      \begin{tabular}{|c|c|c|}
        \hline
        $Z\otimes X$ & $X \otimes Z$ & $Y \otimes Y$ \\ \hline
        $\mathbbm{1}\otimes X$ & $X \otimes \mathbbm{1}$ & $X \otimes X$ \\ \hline
        $Z\otimes \mathbbm{1}$ & $\mathbbm{1}\otimes Z$ & $Z \otimes Z$ \\
        \hline
      \end{tabular}
      \\ \\
      
      % Bottom-left MSG
      \begin{tabular}{|c|c|c|}
        \hline
        $\mathbbm{1}\otimes X$ & $X \otimes \mathbbm{1}$ & $X \otimes X$ \\ \hline
        $Z\otimes X$ & $X \otimes Z$ & $Y \otimes Y$ \\ \hline
        $Z\otimes \mathbbm{1}$ & $\mathbbm{1}\otimes Z$ & $Z \otimes Z$ \\
        \hline
      \end{tabular}
      &
      % Bottom-right MSG
      \begin{tabular}{|c|c|c|}
        \hline
        $\mathbbm{1}\otimes Z$ & $X \otimes \mathbbm{1}$ & $X \otimes Z$ \\ \hline
        $Z\otimes \mathbbm{1}$ & $\mathbbm{1}\otimes Y$ & $Z \otimes Y$ \\ \hline
        $Z\otimes Z$ & $X \otimes Y$ & $Y \otimes X$ \\
        \hline
      \end{tabular}
      
    \end{tabular}
    \caption{An example of $\bm{\mathcal{G}}^{\text{MSG}}$.}
    \end{table}
\end{example}

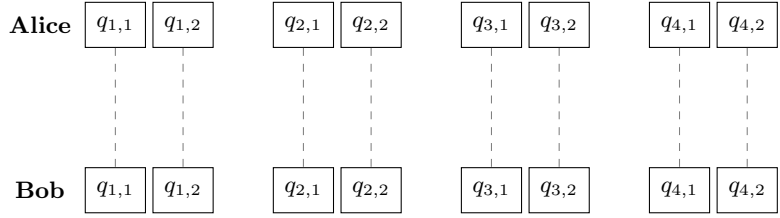
\begin{figure}[ht]\label{fig: parallel}
    \centering
    \begin{tikzpicture}[every node/.style={font=\small},
      qubit/.style={draw, rectangle, minimum width=0.8cm, minimum height=0.6cm},
      pairbox/.style={draw, fill=blue!15, rounded corners=3pt, minimum width=1.8cm, minimum height=1.1cm},
      arrow/.style={->, thick}]    
        % Alice qubits (top row)
        \node[qubit] (aE1) at (0.0, 2.2) {$q_{1,1}$};
        \node[qubit] (aE2) at (0.9, 2.2) {$q_{1,2}$};
        
        \node[qubit] (aF1) at (2.5, 2.2) {$q_{2,1}$};
        \node[qubit] (aF2) at (3.4, 2.2) {$q_{2,2}$};
        
        \node[qubit] (aG1) at (5.0, 2.2) {$q_{3,1}$};
        \node[qubit] (aG2) at (5.9, 2.2) {$q_{3,2}$};
        
        \node[qubit] (aH1) at (7.5, 2.2) {$q_{4,1}$};
        \node[qubit] (aH2) at (8.4, 2.2) {$q_{4,2}$};
        
        % Bob qubits (bottom row)
        \node[qubit] (bE1) at (0.0, 0.0) {$q_{1,1}$};
        \node[qubit] (bE2) at (0.9, 0.0) {$q_{1,2}$};
        
        \node[qubit] (bF1) at (2.5, 0.0) {$q_{2,1}$};
        \node[qubit] (bF2) at (3.4, 0.0) {$q_{2,2}$};
        
        \node[qubit] (bG1) at (5.0, 0.0) {$q_{3,1}$};
        \node[qubit] (bG2) at (5.9, 0.0) {$q_{3,2}$};
        
        \node[qubit] (bH1) at (7.5, 0.0) {$q_{4,1}$};
        \node[qubit] (bH2) at (8.4, 0.0) {$q_{4,2}$};
        
        % Entanglement lines (dashed)
        \draw[gray, dashed] (aE1.south) -- (bE1.north);
        \draw[gray, dashed] (aE2.south) -- (bE2.north);
        
        \draw[gray, dashed] (aF1.south) -- (bF1.north);
        \draw[gray, dashed] (aF2.south) -- (bF2.north);
        
        \draw[gray, dashed] (aG1.south) -- (bG1.north);
        \draw[gray, dashed] (aG2.south) -- (bG2.north);
        
        \draw[gray, dashed] (aH1.south) -- (bH1.north);
        \draw[gray, dashed] (aH2.south) -- (bH2.north);
        
        % Labels on the left
        \node at (-1.0, 2.2) {\textbf{Alice}};
        \node at (-1.0, 0.0) {\textbf{Bob}};    
    \end{tikzpicture}
    \caption{For $K=4$, this graphic depicts 4 parallel and independent games, where each game requires 2 entangled states. For game $i$ and qubit $j$, the $q_{i,j}$ denotes qubits while the dashed line represents the entanglement between the qubits (players).}
\end{figure}

    \begin{theorem}
        Let $\{\mathcal{G}_i\}_{i=1}^K$ be $K$ non-local games, each admitting a perfect quantum strategy on a fixed number of qubits $n_i$ for each player. 
        Let $\mathcal{A}_i$ be the $C^*$-algebra generated by the measurement operators of $\mathcal{G}_i$ for Alice, and $\mathcal{B}_i$ the corresponding algebra for Bob. 
        Suppose that for each game $\mathcal{G}_i$ there exists a shared entangled state $\ket{\psi_i} \in \mathbb{C}^{2^{n_i}} \otimes \mathbb{C}^{2^{n_i}}$
        and measurement operators $\{M_{x_i,a_i}\}_a \subset \mathcal{A}_i$ and $\{N_{y_i,b_i}\}_b \subset \mathcal{B}_i$ 
        achieving the winning conditions of $\mathcal{G}_i$ with probability $1$.  
        
        For each round, the referee samples $i\in\{1, \dots, K\}$ uniformly, sends questions $(x_i,y_i)$ for game $\mathcal{G}_i$, and accepts if only if game $\mathcal{G}_i$ is accepted. Let $d=\max_i 2^{n_i}$. Then there exists a quantum strategy on $\bar{n}=\max_i n_i$ qubits per player that wins the selected game with probability 1. 
    \end{theorem}

    \begin{proof}
        Let $\mathcal{H}_D=\mathbb{C}^d$ be each player's local Hilbert space and let $\ket{\Phi_d}\in \mathcal{H}_D\otimes \mathcal{H}_D$ be a maximally entangled state. For each game $\mathcal{G}_i$, fix isometries 
        \begin{equation}
            U_{A_i}: \mathcal{H}_{A_i}\hookrightarrow \mathcal{H}_D \quad\text{and}\quad U_{B_i}: \mathcal{H}_{B_i}\hookrightarrow \mathcal{H}_D.
        \end{equation}
        Define POVMs on $\mathcal{H}_D$ by 
        \begin{equation}
            \widehat{M}_{x_i,a_i}:=U_{A_i}\,M_{x_i,a_i}\,(U_{A_i})^* \quad\text{and}\quad \widehat{N}_{y_i,b_i}:=U_{B_i}\,N_{y_i,b_i}\,(U_{B_i})^*,
        \end{equation}
        which are valid POVMs because isometries preserve positivity and completeness. 

        With this strategy, playing the game is as follows: upon receiving questions $(x_i,y_i)$, Alice and Bob use the $i$-th POVM families $\{\widehat{M}_{x_i,a_i}\}_{a_i}$ and $\{\widehat{N}_{y_i,b_i}\}_{b_i}$ on the shared state $\ket{\Phi_d}$. The outcome distribution matches exactly that of the original perfect strategy for game $\mathcal{G}_i$, hence the referee accepts with probability $1$ for every question pair. More concretely, for each game $\mathcal{G}_i$, inputs $(x_i,y_i)$, and outcomes $(a_i,b_i)$, let $P^{\mathrm{new}}_i(a_i,b_i\,|\,x_i,y_i) := \big\langle \Phi_d \big|\, \widehat M_{x_i,a_i}\otimes \widehat N_{y_i,b_i} \,\big| \Phi_d \big\rangle$ and $P^{\mathrm{orig}}_i(a_i,b_i\,|\,x_i,y_i):=\big\langle \psi_i \big|\, M_{x_i,a_i}\otimes N_{y_i,b_i} \,\big| \psi_i \big\rangle $. Then, we obtain 
        \begin{equation}\label{eq: prob match} P^{\mathrm{new}}_i(a_i,b_i\,|\,x_i,y_i)= P^{\mathrm{orig}}_i(a_i,b_i\,|\,x_i,y_i).
    \end{equation}
    Let $\lambda_i(x_i,y_i,a_i,b_i)\in\{0,1\}$ be the standard rule function and write 
    \begin{equation}
        \sum_{a_i,b_i} \lambda_i(x_i,y_i,a_i,b_i)\,P^{\mathrm{orig}}_i(a_i,b_i\,|\,x_i,y_i)=1.
    \end{equation}
    Using \eqref{eq: prob match},
    \begin{equation}
        \sum_{a_i,b_i} \lambda_i(x_i,y_i,a_i,b_i)\,P^{\mathrm{new}}_i(a_i,b_i\,|\,x_i,y_i) =\sum_{a_i,b_i} \lambda_i(x_i,y_i,a_i,b_i)\,P^{\mathrm{orig}}_i(a_i,b_i\,|\,x_i,y_i) =1.
    \end{equation}
    i.e., the referee accepts with probability 1 for every input pair. Since the local Hilbert space was $\mathbb{C}^d$, it follows that this strategy uses $\bar{n}$ qubits per player, completing the proof. 
    \end{proof}

    \begin{remark}
        This result highlights a single strategy that can pass any of the $K$ games chosen with provable minimal qubits. In practice, one would use this result to prepare a single entangled state rather than many states for different games. Whatever game $i$ is chosen, one would reroute to that game `block' and win. The advantage to this strategy is that it does not require re-preparing states, no swapping of circuits, or no hardware reconfiguration per game.
    \end{remark}

    \begin{example}
        We take 2 NLGs, each with perfect strategies and embed them into a common Hilbert space. Let $\mathcal G_1$ be the Mermin magic square game (MSG) and $\mathcal G_2$ the magic rectangle game (MRG)~\cite{Mermin1990b}. Each game admits a perfect quantum strategy on $n_1=2$ and $n_2=3$ qubits per player, respectively. We construct a common strategy on $\bar n=\max\{n_1,n_2\}=3$ qubits per player that perfectly plays either game.

        Let each player's local Hilbert space be
        \begin{equation}
            \mathcal H_D = \mathbb C^d \cong (\mathbb C^2)^{\otimes 3}, \quad\text{where}\quad d = 2^{\bar n}=8.
        \end{equation}
        The players share the maximally entangled state
        \begin{equation}
            \ket{\Phi_8}
                = \frac{1}{\sqrt{8}} \sum_{j=0}^7 \ket{j}\otimes\ket{j}
                = \ket{\Phi_2}^{\otimes 3},\quad\text{where}\quad
                \ket{\Phi_2} = \tfrac{1}{\sqrt2}(\ket{00}+\ket{11}).
        \end{equation}
        Label Alice’s qubits $(q_1,q_2,q_3)$ and Bob’s $(q_1',q_2',q_3')$.
    
        If the referee samples and selects the MRG, play the game as usual. Since this game requires the maximal amount of qubits, the embeddings for this game are identities. Hence $\widehat{M}_{x,a}=M_{x,a}$ and $\widehat{N}_{y,b}=N_{y,b}$.
    
        If the referee samples and selects the MSG, embed $\mathcal H_{A_1}$ and $\mathcal H_{B_1}$ into $\mathcal H_D$ via
        \begin{equation}
            U_{A_1}(\ket{\alpha})=\ket{\alpha}\otimes\ket{0},
            \quad\text{and}\quad
            U_{B_1}(\ket{\beta})=\ket{\beta}\otimes\ket{0}.
        \end{equation}
        Then, the lifted POVMs are the same with an additional tensor product with the identity operator making $\widehat{M}_{x_1,a_1} = M_{x_1,a_1}\otimes \mathbbm{1}$ and similarly for $\widehat{N}_{y_1,b_1}$. 
    \end{example}

\section{Cartan Extrapolation}\label{sec: Cartan}
    This section develops a Lie-algebraic lens for identifying and exploring structures across multiple games. The key idea is that Cartan decompositions of $\mathfrak{su}(2^n)$ are recursive~\cite{khaneja2001cartan} and this exposes canonical, `interaction sectors' in which families of measurement algebras can be simultaneously aligned. Moreover, since dimensions of Lie algebras are used as a proxy for expressiveness~\cite{larocca2022group}, we can use this line of thought to understand the entanglement reduction from the Cartan perspective.

\subsection{Cartan preliminaries}
    Recall $\mathfrak{su}(2^n)$ is a Lie algebra comprised of $2^n\times 2^n$ traceless skew-Hermitian matrices. 
    A \emph{Cartan decomposition} of a compact real semisimple Lie algebra $\mathfrak{g}$ is a vector space decomposition
    \begin{equation}\label{eq:Cartan-decomp}
        \mathfrak{g}\;=\;\mathfrak{k}\  \oplus\ \mathfrak{p}
    \end{equation}
   and bracket relations
    \begin{equation*}
        [\mathfrak{k},\mathfrak{k}]\subseteq \mathfrak{k},\qquad
        [\mathfrak{k},\mathfrak{p}]\subseteq \mathfrak{p},\qquad
        [\mathfrak{p},\mathfrak{p}]\subseteq \mathfrak{k}.
    \end{equation*}
    Here, $\mathfrak{k}$ is the Lie subalgebra fixed by an involution $\theta:\mathfrak{g}\to\mathfrak{g}$ so that 
    \begin{equation*}
        \mathfrak{k}= \{X\in \mathfrak{g}: \theta(X)=X \}.
    \end{equation*}
    From this, then $\mathfrak{p}$ is simply the $-1$ eigenspace of $\theta$. That is 
    \begin{equation}
        \mathfrak{p}= \{X\in \mathfrak{g}: \theta(X)= -X \}.
    \end{equation}
    A \emph{Cartan subalgebra} (CSA) $\mathfrak{h}\subset\mathfrak{g}$ is a maximal abelian subalgebra consisting of semisimple elements. For a symmetric pair $(\mathfrak{g},\mathfrak{k})$, a \emph{maximal abelian} subspace $\mathfrak{a}\subset\mathfrak{p}$ plays a central role in the Cartan ($KAK$) factorization. Namely, 
    \begin{equation}\label{eq:KAK}
        G\;=\;K\,\exp(\mathfrak{a})\,K, \quad\text{where}\quad K:=\exp(\mathfrak{k}) \quad\text{and}\quad G:=\exp(\mathfrak{g}).
    \end{equation}
    Intuitively, the component $\exp(\mathfrak{a})$ captures the non-local (or ``interaction'') degrees of freedom.
    \begin{example}
        For $\mathfrak{g}=\mathfrak{su}(4)$, the standard Cartan decomposition aligned with the local subgroup $K \cong \mathrm{SU}(2)\otimes \mathrm{SU}(2)$ yields
    \begin{equation}\label{eq:su4-cartan}
        \mathfrak{k}\;=\;\mathrm{span}\{X\otimes \mathbbm{1},\ Y\otimes \mathbbm{1},\ Z\otimes \mathbbm{1},\ \mathbbm{1}\otimes X,\ \mathbbm{1}\otimes Y,\ \mathbbm{1}\otimes Z\},\quad
        \mathfrak{a}\;=\;\mathrm{span}\{X\otimes X,\ Y\otimes Y,\   Z\otimes Z\},
    \end{equation}
    where $X,Y,Z$ are the 2X2 Pauli operators. 
    Any $U\in \mathrm{SU}(4)$ admits $U=(k_1\otimes k_2)\,\exp(c_x X\otimes X+c_y Y\otimes Y+c_z Z\otimes Z)\,(k_3\otimes k_4)$, with canonical parameters $(c_x,c_y,c_z)$. This separates “local rotations” from the three entangling directions. See~\cite{khaneja2001cartan} for further details on how to recursively construct $\mathfrak{su}(2^n)$ from this example. 
    \end{example}

    Recall the adjoint operator $\mathrm{Ad}$ definition. 
    \begin{definition}
        Let $G$ be a Lie group with Lie algebra $\mathfrak{g}$. The adjoint representation of $G$ on $\mathfrak{g}$ is defined as
        \begin{equation}
            \mathrm{Ad}_g(X)\coloneqq \left.\frac{d}{dt}\right|_{t=0} ge^{tX}g^{-1} = gXg^{-1},
        \end{equation}
        for $g\in G$ and $X\in\mathfrak{g}$. When $G=SU(2^n)$, this reduces to 
        \begin{equation}
            \mathrm{Ad}_U(X) = UXU^\dagger.
        \end{equation}
        \end{definition}
    For a deeper understanding of Lie algebras and their Cartan decompositions, we refer the reader to \cite{hall2013lie}. 
    
    \subsection{Game algebras and Cartan sectors}
    Let $\mathcal{A}_i$ and $\mathcal{B}_i$ be the C*-algebras generated by the perfect strategy measurements of game $\mathcal{G}_i$, and let their skew-Hermitian Lie closures (i.e., the Lie algebra generated by the game observables) be denoted as 
    \begin{equation}\label{eq: game lie algebra}
        \mathfrak{g}_i^A = \mathrm{Lie}\{iM\, : M\in\mathcal{A}_i, \, M^*=M\}\subseteq\mathfrak{su}(2^{n_i}),
    \end{equation}
    and define $\mathfrak{g}_i^B$ similarly. Note that $\mathrm{Lie}\{\cdot\}$ here denotes the smallest Lie algebra containing the given set, which is equivalent to the set closed under commutators and linear combinations. See~\cite{hall2013lie} for more details on Lie constructions. 

    After embedding into a larger $\mathbb{C}^{2^n}$ space as in \Cref{sec: dependent}, we consider embeddings by assuming there exists $n$ and injective $*$-homomorphisms (commuting embeddings)
    \begin{equation}\label{eq: embeddings}
        \iota_i^A:\mathcal{A}_i\hookrightarrow \mathcal{B}(\mathbb{C}^{2^n}) \quad\text{and}\quad \iota_i^B:\mathcal{B}_i\hookrightarrow \mathcal{B}(\mathbb{C}^{2^n}),
    \end{equation}
    such that for all $i\neq j$ the images commute,
    $[\iota_i^A(\mathcal{A}_i),\iota_j^A(\mathcal{A}_j)]=0$ and
    $[\iota_i^B(\mathcal{B}_i),\iota_j^B(\mathcal{B}_j)]=0$.
    Such an assumption allows these finite-dimensional game algebras to embed into a single matrix algebra, chosen so that the images for different games commute.  
    From this, we then consider $\iota_i^A(\mathfrak{g}_i^A),\iota_i^B(\mathfrak{g}_i^B)\subseteq \mathfrak{su}(2^n)$. 

    The following results provide an algebraic certification for when operator algebras arising from multiple NLGs can be realized inside a single Hilbert space with commuting embeddings. In particular, it shows that if each game Lie algebra can be aligned into a common Cartan decomposition, then all the games can be played in parallel on $n$ qubits per player. This connects the structural properties of $\mathfrak{su}(2^n)$ and its Cartan decomposition with the existence of jointly measurable products of POVMs and the shared entangled state. 
    \begin{theorem}\label{theorem: embedding existence}        
        Let $\iota_i^A:\mathcal A_i\hookrightarrow\mathcal B(\mathbb C^{2^n})$ and  $\iota_i^B:\mathcal B_i\hookrightarrow\mathcal B(\mathbb C^{2^n})$ be the embeddings in \Cref{eq: embeddings}. Write $\iota_i^A(\mathfrak g_i^A),\iota_i^B(\mathfrak g_i^B)\subset\mathfrak{su}(2^n)$ for the corresponding embedded Lie algebras. Suppose there exists a Cartan decomposition $\mathfrak{su}(2^n)=\mathfrak{k}\oplus\mathfrak{p}$ with a maximal abelian subspace  $\mathfrak{a}\subset \mathfrak{p}$ and elements $k_i^\mathrm{A},k_i^\mathrm{B}\in K$ from \Cref{eq:KAK} such that for all $i$,
        \begin{equation}\label{eq:cartan-align}     
            \mathrm{Ad}_{k_i^\mathrm{A}}\left(\iota_i^A(\mathfrak{g}_i^A)\right)\ \subseteq\ \mathfrak{c}_A\oplus \mathrm{span}\,\mathfrak{a}\eqqcolon \mathfrak{k}_A,\qquad \mathrm{Ad}_{k_i^\mathrm{B}}\left(\iota_i^B(\mathfrak{g}_i^B)\right)\ \subseteq\ \mathfrak{c}_B\oplus \mathrm{span}\,\mathfrak{a}\eqqcolon \mathfrak{k}_B,
        \end{equation}
        where $\mathfrak{c}_A,\mathfrak{c}_B$ are fixed abelian subalgebras of $\mathfrak{su}(2^n)$ independent of $i$.
        Then, there exists an injective $*$-homomorphism 
        \begin{equation}
            \widehat{\iota}_i^A: \mathcal{A}_i \hookrightarrow \mathcal{B}(\mathbb{C}^{2^n})\quad\text{and}\quad \widehat{\iota}_i^B:\mathcal{B}_i\hookrightarrow \mathcal{B}(\mathbb{C}^{2^n})
        \end{equation}
        such that their images lie in the abelian $C*$-algebras $\mathcal{T}_A\coloneqq C^*(\mathrm{exp}(\mathfrak{k}_A))$ and $\mathcal{T}_B\coloneqq C^*(\mathrm{exp}(\mathfrak{k}_B))$, hence commute element-wise across all $i$. 
    \end{theorem}

    \begin{proof}
        Since $\mathfrak{c}_A$ and $\mathfrak{c}_B$ are abelian, it follows that $\mathfrak{c}_A\oplus \mathrm{span}\,\mathfrak{a}$ and $\mathfrak{c}_B\oplus \mathrm{span}\,\mathfrak{a}$ are also abelian. This follows from how the bracket is defined for direct sums. i.e., for any $h_1,h_2\in \mathfrak{c}_A$ and any $c_1,c_2\in \mathrm{span}\,\mathfrak{a}$, it follows that $[(h_1,c_1),(h_2,c_2)]=\left([h_1,h_2],[c_1,c_2]\right)$. Exponentiation of an abelian set preserves the commutativity structure, thus $\mathcal{T}_A$ and $\mathcal{T}_B$ are abelian. Let 
        \begin{equation}
            \widehat\iota_i^A := \mathrm{Ad}_{k_i^{\mathrm A}}\circ \iota_i^A \quad\text{and}\quad \widehat\iota_i^B \;:=\; \mathrm{Ad}_{k_i^{\mathrm B}}\circ \iota_i^B.
        \end{equation}
        Since $\iota_A$ and $\iota_B$ are already injective $*$-homomorphisms, it follows that conjugation by a unitary $k_i^A$ or $k_i^B$ preserves multiplication, adjoints, and norms. Thus $\widehat\iota_i^A$ and $\widehat\iota_i^B$ are injective $*$-homomorphisms with $\mathrm{range}(\widehat\iota_i^A)\subset\mathcal{T}_A$ and $\mathrm{range}(\widehat\iota_i^B)\subset\mathcal{T}_B$. This concludes that for $i\neq j$, these embeddings also commute, that is $[\widehat\iota_i^A(\mathfrak{g}^A_i),\widehat\iota_i^B(\mathfrak{g}^B_j)]=0$. This proves the claimed commuting embeddings across all games.
    \end{proof}

    The following result guarantees a qubit reduction.  

    \begin{theorem}\label{theorem: qubit reduction}
        As before, let $\{\mathcal{G}_i\}_{i=1}^K$ be a collection of K non-local games, each with perfect strategies requiring $n_i$ qubits. Let $\mathfrak{g}_i^A$ and $\mathfrak{g}_i^B$ be the Lie algebras generated from the game C*algebras as in \Cref{eq: game lie algebra} with the embeddings $\iota_i^A(\mathfrak g_i^A),\iota_i^B(\mathfrak g_i^B)\subset \mathfrak{su}(2^n)$. Let $k_i^{\mathrm A},k_i^{\mathrm B}\in SU(2^n)$ be such that the images $\mathrm{Ad}_{k_i^{\mathrm A}}\big(\iota_i^A(\mathfrak g_i^A)\big)$ and $\mathrm{Ad}_{k_i^{\mathrm B}}\big(\iota_i^B(\mathfrak g_i^B)\big)$ satisfy 
        \begin{equation}\label{eq: big Lie assumption}            
            \mathrm{dim}\left(\mathrm{span}\left\{\mathrm{Ad}_{k_i^{\mathrm A}}\big(\iota_i^A(\mathfrak g_i^A)\big)\right\}_i\right)<r_A\quad\text{and}\quad \mathrm{dim}\left(\mathrm{span}\left\{\mathrm{Ad}_{k_i^{\mathrm B}}\big(\iota_i^B(\mathfrak g_i^B)\big)\right\}_i\right)<r_B,
        \end{equation}
        where $r_A:=\mathrm{dim}(\mathfrak k_A)$ and $r_B:=\mathrm{dim}(\mathfrak k_B)$, for some abelian subalgebras $\mathfrak k_A,\mathfrak k_B\subset \mathfrak{su}(2^n)$. Assume for the given qubit counts per game $n_1,\dots,n_K>1$ we have that 
        \begin{equation}\label{eq: max rarb}
            \max\{r_A,r_B\}<\sum\limits_{i=1}^K (2^{n_i}-1).
        \end{equation}
        Let $N\coloneqq\sum\limits_{i=1}^K n_i$ and let $n=\min\{m\in\mathbb{N}: 2^m-1\geq r\}$, where $r=\max\{r_A,r_B\}$. If $K\geq 2$, then $n<N$.
    \end{theorem}
    
    \begin{proof}
        Recall the rank of $\mathfrak{su}(2^n)$ is $2^n-1$. This implies that any abelian subalgebra has dimension $\leq 2^n-1$, and hence $r\leq 2^n-1$. By the definition of $n$, it follows that 
        \begin{equation}
            2^{n-1}-1\leq r\leq 2^n-1 \iff 2^{n-1}\leq r+1\leq 2^n. 
        \end{equation}
        
        We claim that 
        \begin{equation}
            \sum\limits_{i=1}^K2^{n_i}\leq 2^{\sum\limits_{i=1}^Kn_i}=2^N,
        \end{equation}
        and we show this via induction. For the base case, when $K=1$, there is equality. If we assume the hypothesis is true for $K-1$ and that 
        \begin{equation*}
            \sum\limits_{i=1}^{K-1}2^{n_i}\leq 2^{\sum\limits_{i=1}^{K-1}n_i},
        \end{equation*}
        since $n_K\geq 1$, it follows that 
        \begin{align}
            2^N =  2^{\sum\limits_{i=1}^Kn_i} 
            & = \left(2^{\sum\limits_{i=1}^{K-1}n_i}\right)2^{n_K} \\
            & \geq  \left(\sum\limits_{i=1}^{K-1}2^{n_i}\right)2^{n_K} \\
            & \geq \left(\sum\limits_{i=1}^{K-1}2^{n_i}\right) + 2^{n_K} \\
            & = \sum\limits_{i=1}^{K}2^{n_i}.
        \end{align}
        Then in particular, if $K\geq 2$, then we also have 
        \begin{align}
            \sum\limits_{i=1}^K\left(2^{n_i}-1\right) +1 
            & = \sum\limits_{i=1}^K 2^{n_i} - K+1 \\ 
            & \leq 2^N - K+1 \\
            & < 2^N, \quad\text{whenever } K\geq 2.
        \end{align}
        By assumption, $r\leq \sum\limits_{i=1}^K(2^{n_i}-1)$ implies $r+1\leq \sum\limits_{i=1}^K(2^{n_i}-1) +1$. So by monotonicity of the $\log$ function, we have 
        \begin{align}
            n = \log_2(r+1) 
            & \leq \log_2\left(\sum\limits_{i=1}^K\left(2^{n_i}-1\right) +1 \right) \\
            & < \log_2(2^N) = N,
        \end{align}
        which completes the proof. 
    \end{proof}

    The assumption from \Cref{eq: big Lie assumption} is essential for the dimension reduction argument. It guarantees that the collection of Lie algebras Cartan components coming from the players fit inside a fixed abelian subalgebra of dimension $r_A$ (resp.\ $r_B$). Such an alignment into the same abelian subalgebra of bounded dimension produces an overlap, which causes the reduction. 

    In what comes next is a sequence of results to guarantee the state used in \Cref{theorem: all games} is guaranteed to exist and be entangled. Fix embeddings $\iota_i^A:\mathcal{A}_i\hookrightarrow\mathcal{B}(\mathbb{C}^{2^n})$ and  $\iota_i^B:\mathcal{B}_i\hookrightarrow\mathcal{B}(\mathbb{C}^{2^n})$ as in~\Cref{eq: embeddings}, with the pairwise commuting ranges across $i$. For each game $\mathcal{G}_i$ and question pair $(x_i,y_i)$, define the \emph{acceptance operator}
    \begin{equation}
        W_i(x_i,y_i)\coloneqq\sum_{\substack{(a_i,b_i):\\\lambda(x_i,y_i,a_i,b_i)=1}} \iota_i^A\big(M_{x_i,a_i}\big)\ \otimes\ \iota_i^B\big(N_{y_i,b_i}\big).
    \end{equation}
    Each $W_i(x_i,y_i)$ is a positive and equals the identity on any perfect strategy state for $\mathcal{G}_i$.

    \begin{remark}
        Under~\Cref{theorem: embedding existence}, the Lie algebra per game aligns into a common abelian $C^*$-algebra for each player. In that Cartan space, all $W_i(x_i,y_i)$ are simultaneously diagonal, so each is a projector onto a set of basis indices. Let $\{|e_j^A\rangle\}_{j=1}^{d}$ and $\{|e_\ell^B\rangle\}_{\ell=1}^{d}$ be the corresponding orthonormal eigenbases. 
    \end{remark}

    \begin{definition}[Common Winning Sector (CWS)]
        Let $\mathcal{S}_{i,(x_i,y_i)}$ be the set of diagonal basis indices on which $W_i(x_i,y_i)$ has eigenvalue $1$. The \emph{Common Winning Sector} is
        \begin{equation}
            \mathrm{CWS}\coloneqq\bigcap_{i=1}^K\ \bigcap_{(x_i,y_i)} \mathcal{S}_{i,(x_i,y_i)}.
        \end{equation}
    \end{definition}
    The following proves when the CWS is nonempty. Equivalently, $\mathrm{CWS}\neq\emptyset$ iff $\bigcap\limits_{i,(x_i,y_i)} \mathrm{Ran}\,W_i(x_i,y_i)\neq\{0\}$. 
    \begin{lemma}\label{lemma: nonempty CWS}
        If $\iota_i^{A/B}$ act on disjoint tensor factors, then for perfect per-game states $|\psi_i\rangle$, the product $|\Psi\rangle=\bigotimes\limits_i|\psi_i\rangle$ lies in $\bigcap\limits_{i,(x_i,y_i)}\mathrm{Ran}\,W_i(x_i,y_i)$.
    \end{lemma}

    \begin{proof}
         Fix $i\in\{1,\dots,K\}$ and a question pair $(x_i,y_i)$ of game $\mathcal{G}_i$. Under the disjoint-factor hypothesis, there exist positive operators 
         \begin{equation}
             \widetilde{W}_i(x_i,y_i) \coloneqq \sum\limits_{\substack{(a_i,b_i): \\ \lambda(x_i,y_i,a_i,b_i)=1}} M_{x_i,a_i}\otimes N_{y_i,b_i} \quad\text{on}\quad \mathcal{H}_{A_i}\otimes\mathcal{H}_{B_i}
         \end{equation}
         such that 
         \begin{equation}
             W_i(x_i,y_i) = \Big(\bigotimes_{j<i} \mathbbm{1}_{A_j}\otimes \mathbbm{1}_{B_j}\Big)  \otimes  \widetilde{W}_i(x_i,y_i)  \otimes \Big(\bigotimes_{j>i} \mathbbm{1}_{A_j}\otimes \mathbbm{1}_{B_j}\Big).
         \end{equation}
         Since $\{M_{x_i,a_i}\},\{N_{y_i,b_i}\}$ together with $|\psi_i\rangle$ form a perfect strategy for $\mathcal{G}_i$, we have for every $(x_i,y_i)$ that 
         \begin{equation}
             \big\langle \psi_i\big| \widetilde{W}_i(x_i,y_i)\,\big|\psi_i\big\rangle \;=\; 1, \qquad 0 \;\le\; \widetilde{W}_i(x_i,y_i) \;\le\; \mathbbm{1}.
         \end{equation}
         As before, let $X:=\mathbbm{1}-\widetilde{W}_i(x_i,y_i)\ge 0$. Then
         \begin{equation}
             \|X^{\frac{1}{2}}\ket{\psi_i}\|^2 = \langle\psi_i|X|\psi_i\rangle = 1- \langle\psi_i|\widetilde{W}_i(x_i,y_i)|\psi_i\rangle = 0. 
         \end{equation}
         Hence $X^{\frac{1}{2}}\ket{\psi_i}=0$ and thus $X\ket{\psi_i}=0$. It follows that $\widetilde{W}_i(x_i,y_i)|\psi_i\rangle =\ket{\psi_i}$. 

         Consider a global product state $\ket{\Psi}=\bigotimes\limits_{j=1}^K\ket{\psi_j}$. Using the tensor form of $W_i(x_i,y_i)$, we obtain
         \begin{equation}
             W_i(x_i,y_i)\ket{\Psi} = \left(\bigotimes\limits_{j<i}\ket{\psi_j} \right) \otimes (\widetilde{W}_i(x_i,y_i)|\psi_i\rangle) \otimes \left(\bigotimes\limits_{j>i}\ket{\psi_j} \right) =\ket{\Psi}.
         \end{equation}
         Thus $|\Psi\rangle\in \mathrm{Ran}(W_i(x_i,y_i))$, for every $(x_i,y_i)$ of game $\mathcal{G}_i$. Since this holds for each $i=1,\dots,K$, we conclude
         \begin{equation}
             |\Psi\rangle \in \bigcap_{i=1}^K \bigcap_{(x_i,y_i)} \mathrm{Ran}(W_i(x_i,y_i)),
         \end{equation}
         proving CWS is nonempty. 
    \end{proof}
    Thus $(j,\ell)\in\mathrm{CWS}$ iff $|e_j^A\rangle\otimes|e_\ell^B\rangle$ lies in the $+1$ eigenspace of every $W_i(x_i,y_i)$. In particular, each $W_i(x_i,y_i)$ acts as the identity operator on $\ket{\Psi}$ due to the commutativity across all $i$. Hence, any unit vector $\ket{\Psi}\in\bigcap\limits_{i,(x_i,y_i)}\mathrm{Ran}(W_i(x_i,y_i))$ satisfies
    \begin{equation}\label{eq: state from CWS}
        \Big\langle \Psi\ \Big|\ \prod_{i=1}^K W_i(x_i,y_i)\ \Big|\ \Psi\Big\rangle = 1 \quad\text{for all}\ (x_1,y_1),\dots,(x_K,y_K).
    \end{equation}
    % product of id operators acting on state still id operator.
    Finally, we verify that the chosen state is indeed entangled, as this is necessary for advantageous NLGs. 
    
    \begin{lemma}\label{lemma: entangled-CWS}
        Suppose $\mathrm{CWS}$ contains at least two distinct accepted pairs of indices $\{(k_j,\ell_j)\}_{j=1}^{L}$ with $L\ge 2$. Define, in the diagonal space,
        \begin{equation}
            |\Phi\rangle \coloneqq \frac{1}{\sqrt{L}}\sum_{j=1}^{L} |e_{k_j}^A\rangle\otimes|e_{\ell_j}^B\rangle.
        \end{equation}
        Then $|\Phi\rangle$ is entangled whenever the set $\{|e_{k_j}^A\rangle\}$ (equivalently $\{|e_{\ell_j}^B\rangle\}$) contains at least two orthonormal vectors.
    \end{lemma}

    \begin{proof}
        Write $|\Phi\rangle=\sum\limits_j \alpha_j\,|e_{k_j}^A\rangle\otimes|e_{\ell_j}^B\rangle$ with $\alpha_j=L^{-1/2}$. The reduced state on Alice is $\rho_A=\mathrm{Tr}_B(|\Phi\rangle\!\langle\Phi|)=\frac{1}{L}\sum\limits_j |e_{k_j}^A\rangle\!\langle e_{k_j}^A|$, which is mixed (rank $\ge 2$) whenever at least two distinct $k_j$ appear. Hence $|\Phi\rangle$ has Schmidt rank $\ge 2$ and is entangled.
    \end{proof}
    
\section{Compressed Games}\label{sec: dependent}
    In this section, we utilize shared algebraic structures, via commuting subalgebras and partial isometries, to argue when and how qubits can be reused across games. Such a reduction from the tensor product case above may be interpreted as game compressing.

    As before, let $\{\mathcal{G}_i\}_{i=1}^K$ be $K$ non-local games, each admitting a perfect quantum strategy on a fixed number of qubits $n_i$ for each player. 
    Let $\mathcal{A}_i$ be the $C^*$-algebra generated by the measurement operators of $\mathcal{G}_i$ for Alice, and $\mathcal{B}_i$ the corresponding algebra for Bob. 
    Suppose that for each game $\mathcal{G}_i$ there exists a shared entangled state 
    \begin{equation}
        \ket{\psi_i} \in \mathbb{C}^{2^{n_i}} \otimes \mathbb{C}^{2^{n_i}}
    \end{equation}
    and measurement operators $\{M_{x_i,a_i}\}_{a_i} \subset \mathcal{A}_i$ and $\{N_{y_i,b_i}\}_{b_i} \subset \mathcal{B}_i$ 
    achieving the winning conditions of $\mathcal{G}_i$ with probability $1$.  Define $N=\sum\limits_{i=1}^K n_i$. 
    
    \begin{theorem}\label{theorem: all games}
        Let $\{\mathcal{G}_i\}_{i=1}^K$ be a set of nonlocal games, each admitting a perfect strategy on $n_i$ local qubits per player. For the $i$-th game let $\mathcal{A}_i$ and $\mathcal{B}_i$ be the $C^*$-algebras generated by Alice's and Bob's measurement operators in some perfect strategy, and let $N:=\sum_{i=1}^K n_i$ denote the additive independent maximum number of qubits needed. Assume for some $n\in\mathbb{N}$ there exists injective $*$-homomorphisms
        \begin{equation}
           \iota_i^A:\mathcal{A}_i\hookrightarrow \mathcal{A}\subseteq\mathcal{B}(\mathbb{C}^{2^n}) \quad\text{and}\quad
            \iota_i^B:\mathcal{B}_i\hookrightarrow \mathcal{B}\subseteq\mathcal{B}(\mathbb{C}^{2^n}), 
        \end{equation}
        such that the images pairwise commute across games
        \begin{equation}
            [\iota_i^A(a),\iota_j^A(a')]=0,\quad [\iota_i^B(b),\iota_j^B(b')]=0,
        \end{equation}
        for all $ i\neq j,\ a\in\mathcal{A}_i,\ a'\in\mathcal{A}_j,\ b\in\mathcal{B}_i, \text{and } b'\in\mathcal{B}_j.$
        Then, there exist game algebras $\mathcal{A}$ and $\mathcal{B}$ for Alice and Bob respectively. 

        Furthermore, there exist a shared entangled state $\ket{\Psi}\in\mathbb{C}^{2^n}\otimes\mathbb{C}^{2^n}$ and, for every $i$ and question pair $(x_i,y_i)$ of $\mathcal{G}_i$, along with POVMs $\{\widetilde M_{x_i,a_i}\} \subseteq\mathcal{A}_i$ and $\{\widetilde N_{y_i,b_i}\} \subseteq\mathcal{B}_i$ such that all $K$ games are won using the joint commuting POVMs
        \begin{equation}
            \widetilde M_{\vec{x},\vec{a}}\;:=\;\prod_{i=1}^K \widetilde M_{x_i,a_i} \quad\text{and}\quad \widetilde N_{\vec{y},\vec{b}}\;:=\;\prod_{i=1}^K \widetilde N_{y_i,b_i}
        \end{equation}
        on $\mathbb{C}^{2^n}$ for players Alice and Bob. Such a strategy guarantees perfect winning. i.e., all $K$ games can be won simultaneously on $n<N$ local qubits per player.  
    \end{theorem}

    \begin{proof}
    Fix $K$ games $\{\mathcal{G}_i\}_{i=1}^K$, each admitting a perfect strategy. For each $i$, let $\mathcal{A}_i$ (Alice) and $\mathcal{B}_i$ (Bob) denote the $C^*$-algebras generated by the per-game measurements in some perfect strategy, and let $\{M_{x_i,a_i}\}\subset\mathcal{A}_i$, $\{N_{y_i,b_i}\}\subset\mathcal{B}_i$ be the corresponding POVMs.

   From \Cref{theorem: embedding existence} assume there exists $n$ and unital injective $*$-homomorphisms (commuting embeddings)
    \begin{equation}
        \iota_i^A:\mathcal{A}_i\hookrightarrow \mathcal{B}(\mathbb{C}^{2^n}) \quad\text{and}\quad \iota_i^B:\mathcal{B}_i\hookrightarrow \mathcal{B}(\mathbb{C}^{2^n}),
    \end{equation}
    such that for all $i\neq j$ the images within each player commute,
    $[\iota_i^A(\mathcal{A}_i),\iota_j^A(\mathcal{A}_j)]=0$ and
    $[\iota_i^B(\mathcal{B}_i),\iota_j^B(\mathcal{B}_j)]=0$.
    
    Let $\ket{\Psi}\in\mathbb{C}^{2^n}\otimes\mathbb{C}^{2^n}$ be a bipartite state that realizes perfect play for each game after embedding. Such a state is entangled and exists via \Cref{lemma: entangled-CWS}. That is, for all questions $(x_i,y_i)$,
    \begin{equation}\label{eq: embedded expectation}
    \sum_{\substack{(a_i,b_i):\\ \lambda_i(x_i,y_i,a_i,b_i)=1}}
    \bra{\Psi}\,\iota_i^A\!\big(M_{x_i,a_i}\big)\otimes
                 \iota_i^B\!\big(N_{y_i,b_i}\big)\,\ket{\Psi}=1.
    \end{equation}

    For the joint questions $\vec{x}=(x_1,\dots,x_K)$ to Alice and $\vec{y}=(y_1,\dots,y_K)$ to Bob, define the commuting lifted POVMs
    \begin{equation}
        \widetilde M_{x_i,a_i}:=\iota_i^A\!\big(M_{x_i,a_i}\big)\quad\text{and}\quad \widetilde N_{y_i,b_i}:=\iota_i^B\!\big(N_{y_i,b_i}\big),
    \end{equation}
    and the product measurements for the parallel game as 
    \begin{equation}
        \widetilde M_{\vec{x},\vec{a}}\ :=\ \prod_{i=1}^K \widetilde M_{x_i,a_i}\quad\text{and}\quad  \widetilde N_{\vec{y},\vec{b}}\ :=\ \prod_{i=1}^K \widetilde N_{y_i,b_i},
    \end{equation}
    which are well-defined since the factors commute across $i$.

    For each $i$ and $(x_i,y_i)$ let the POVM element whose expectation equals the accepted (by the referee) probability be 
    \begin{equation}
        \widetilde W_i(x_i,y_i)\ :=\ \sum_{\substack{(a_i,b_i): \\\lambda_i(x_i,y_i,a_i,b_i)=1}} \widetilde M_{x_i,a_i}\otimes \widetilde N_{y_i,b_i}.
    \end{equation}
    By construction, the set $\{\widetilde W_i\}_i$ commute pairwise. From positivity and taking tensor products, we have that $0\le \widetilde W_i\le I$.

    Recall that since $\bra{\Psi} \widetilde W_i\ket{\Psi}=1$ follows from  using $\Cref{eq: embedded expectation}$ and $\| (I-\widetilde W_i)^{1/2}\ket{\Psi}\|^2=\bra{\Psi}(I-\widetilde W_i)\ket{\Psi}=0$, we obtain $\widetilde W_i(x_i,y_i)\ket{\Psi}=\ket{\Psi}$ for all $i$. Since the $\widetilde W_i$ all commute,
    \begin{equation}
        \Big(\prod_{i=1}^K \widetilde W_i(x_i,y_i)\Big)\ket{\Psi}=\ket{\Psi}.
    \end{equation}
   For questions $(\vec{x},\vec{y})$, the POVM that represents all $K$ games is 
    \begin{equation}
        \widetilde W(\vec{x},\vec{y})\ :=\ \sum_{(\vec{a},\vec{b})}
        \widetilde M_{\vec{x},\vec{a}}\otimes \widetilde N_{\vec{y},\vec{b}}
        \ =\ \prod_{i=1}^K \widetilde W_i(x_i,y_i),
    \end{equation}
    hence
    \begin{equation}
        \bra{\Psi}\,\widetilde W(x,y)\,\ket{\Psi}
        =\bra{\Psi}\,\prod_{i=1}^K \widetilde W_i(x_i,y_i)\,\ket{\Psi}
        =1.
    \end{equation}
    Therefore, the product strategy $(\ket{\Psi},\{\widetilde M_{\vec{x},\vec{a}}\},\{\widetilde N_{\vec{y},\vec{b}}\})$ wins the`glued' game with probability $1$ on $n<N$ local qubits per player as desired.
    \end{proof}

To demonstrate this result, we end with an example. 
\begin{example}\label{ex: compressed}
    Let $K=2$ and use 2 independent magic square games (MSG). For each player, split the local Hilbert spaces as 
    \begin{equation}
        \mathcal H_A \cong \underbrace{(\mathbb C^2)^{\otimes 2}}_{\text{data }(q_1,q_2)} \otimes \underbrace{\mathbb C^2}_{\text{control }(c_1)},  \qquad \mathcal H_B \cong \underbrace{(\mathbb C^2)^{\otimes 2}}_{\text{data }(q_4,q_5)} \otimes \underbrace{\mathbb C^2}_{\text{control }(c_2)}.
    \end{equation}
    Let $P_0=\ket{0}\!\bra{0}$ and $P_1=\ket{1}\!\bra{1}$ denote the control projectors (on a single qubit).
    
    \textbf{Acceptance operators:}
    Denote $\mathcal G_1,\mathcal G_2$ as the two (independent) copies of MSG. Each has a perfect strategy on the two-qubit \emph{data} register with POVM elements $\{M_{x_i,a_i}\}_{a_i}, \{N_{y_i,b_i}\}_{b_i}$ and acceptance operator
    \begin{equation}
        W_i(x_i,y_i)\;:=\;\sum_{\substack{a_i,b_i:\\\lambda_i(x_i,y_i,a_i,b_i)=1}} M_{x_i,a_i}\otimes N_{y_i,b_i}, \qquad i\in\{1,2\}.
    \end{equation}
    On the standard two-Bell-pair data state, we have $\langle W_i(x_i,y_i)\rangle=1$ for every $(x_i,y_i)$.

    \textbf{Shared entangled state:} Use 2 Bell pairs on the data qubits and 1 Bell pair on the control qubits:
    \begin{equation}
        \ket{\Psi}\;=\;\big(\ket{\phi^+}\big)_{q_1 q_4}\ \otimes\ \big(\ket{\phi^+}\big)_{q_2 q_5}\ \otimes\ \big(\ket{\phi^+}\big)_{c_1c_2}  \in\ (\mathbb C^{2^3})\otimes(\mathbb C^{2^3}).
    \end{equation}

    \textbf{Commuting embeddings:} Define the embedding POVM elements for Alice with respect to the control:
    \begin{align*}
        \widetilde M_{x_1,a_1} &:= M_{x_1,a_1}\otimes P_0\ +\ \widehat M_{x_1,a_1}\otimes P_1,\\
        \widetilde M_{x_2,a_2} &:= \widehat M_{x_2,a_2}\otimes P_0\ +\ M_{x_2,a_2}\otimes P_1.
    \end{align*}
    Analogously for Bob, we have 
    \begin{align*}
        \widetilde N_{y_1,b_1} &:= N_{y_1,b_1}\otimes P_0\ +\ \widehat N_{y_1,b_1}\otimes P_1,\\
        \widetilde N_{y_2,b_2} &:= \widehat N_{y_2,b_2}\otimes P_0\ +\ N_{y_2,b_2}\otimes P_1.
    \end{align*}
    Here the POVMs $\widehat M_{\cdot,\cdot}$, $\widehat N_{\cdot,\cdot}$ act on the data register and are chosen so that, for every $(x_i,y_i)$,
    \begin{equation}
        \sum_{\substack{a_i,b_i:\\\lambda_i(x_i,y_i,a_i,b_i)=1}}\ \widehat M_{x_i,a_i}\otimes \widehat N_{y_i,b_i}\;=\;\mathbbm{1}_{\text{data}}.
    \end{equation}
    With this construction, each $\widetilde M$ (and $\widetilde N$) is block-diagonal in the same control basis $\{P_0,P_1\}$. In block $P_0$, game~1 is ``active'' while game~2 is ``offline''. In block $P_1$, the roles swap. Hence, the images of the 2 game algebras commute element-wise:
    \begin{equation}
        \big[\widetilde M_{x_1,a_1},\ \widetilde M_{x_2,a_2}\big]=0,
        \qquad
        \big[\widetilde N_{y_1,b_1},\ \widetilde N_{y_2,b_2}\big]=0.     
    \end{equation}

    \textbf{Perfect strategy verification:} Define the embedded acceptance operators as
    \begin{align*}
        \widetilde W_i(x_i,y_i)
        &=
        \sum_{\substack{a_i,b_i:\\\lambda_i(x_i,y_i,a_i,b_i)=1}} \widetilde M_{x_i,a_i}\ \otimes\ \widetilde N_{y_i,b_i} \\
        &= \sum_{\substack{a_i,b_i:\\\lambda_i(x_i,y_i,a_i,b_i)=1}} \Big(M_{x_1,a_1}\otimes P_0\;+\;\widehat M_{x_1,a_1}\otimes P_1\Big) \otimes\
        \Big(N_{y_1,b_1}\otimes P_0\;+\;\widehat N_{y_1,b_1}\otimes P_1\Big)\\
        &=  \sum_{\substack{a_i,b_i:\\\lambda_i(x_i,y_i,a_i,b_i)=1}} \Big(M_{x_1,a_1}\otimes N_{y_1,b_1}\Big)\ \otimes\ \big(P_0\otimes P_0\big)
        \;+\;\sum_{\lambda_1=1}
        \Big(\widehat M_{x_1,a_1}\otimes \widehat N_{y_1,b_1}\Big)\ \otimes\ \big(P_1\otimes P_1\big),
    \end{align*}
    where we recall the cross terms vanish ($P_0P_1=0)$. By our choice of ``offline'' block, we have  
    \begin{equation}
        \sum_{\substack{a_i,b_i:\\\lambda_i(x_i,y_i,a_i,b_i)=1}} \widehat M_{x_1,a_1}\otimes \widehat N_{y_1,b_1} = \mathbbm{1}_{\mathrm{data}},
    \end{equation}
    so we obtain a compact form 
    \begin{equation}\label{eq: w1-embedding}
        \widetilde W_1(x_1,y_1)=
        W_1(x_1,y_1)\ \otimes\ \big(P_0\otimes P_0\big)
        \;+\;\mathbbm{1}_{\mathrm{data}}\ \otimes\ \big(P_1\otimes P_1\big).
    \end{equation}
    Similarly, we have 
    \begin{equation}
     \widetilde W_2(x_2,y_2)=
      \mathbbm{1}_{\mathrm{data}}\ \otimes\ \big(P_0\otimes P_0\big)
      \;+\;W_2(x_2,y_2)\ \otimes\ \big(P_1\otimes P_1\big).
    \end{equation}
    The 2 embedded acceptance operators commute, as they are polynomials in the commuting projectors $P_0$, $P_1$, and on $W_i$ disjoint blocks. In the larger Hilbert space, the full state is
    \begin{equation}
          |\Psi\rangle = |\Phi_{\mathrm{data}}\rangle \otimes |\Phi_{2}\rangle_{\text{control}},\qquad\text{where} \quad |\Phi_{2}\rangle=\tfrac{1}{\sqrt{2}}(\,|00\rangle+|11\rangle\,).
    \end{equation}
    For the control EPR state, we see that 
    \begin{equation}
         (P_0\otimes P_0)\,|\Phi_{2}\rangle = \tfrac{1}{\sqrt{2}}|00\rangle,\qquad
        (P_1\otimes P_1)\,|\Phi_{2}\rangle = \tfrac{1}{\sqrt{2}}|11\rangle,\qquad
        \big(P_0\otimes P_0 + P_1\otimes P_1\big)\,|\Phi_{2}\rangle = |\Phi_{2}\rangle.
    \end{equation}
    Using~\Cref{eq: w1-embedding} and the fact that $W_1|\Phi_{\mathrm{data}}\rangle=|\Phi_{\mathrm{data}}\rangle$,
    we see 
    \begin{align*}
        \widetilde W_1(x_1,y_1)\,|\Psi\rangle
        &= \Big(W_1|\Phi_{\mathrm{data}}\rangle\Big)\otimes\big(P_0\otimes P_0\big)|\Phi_2\rangle
        \;+\;\Big(\mathbbm{1}|\Phi_{\mathrm{data}}\rangle\Big)\otimes\big(P_1\otimes P_1\big)|\Phi_2\rangle\\
        &= |\Phi_{\mathrm{data}}\rangle\ \otimes\ \Big[\,(P_0\otimes P_0)+(P_1\otimes P_1)\,\Big]\,|\Phi_2\rangle
        \;=\;|\Psi\rangle.
    \end{align*}
    Similarly, 
    \begin{equation}
         \widetilde W_2(x_2,y_2)\,|\Psi\rangle \;=\; |\Psi\rangle.
    \end{equation}
    On the controlled Bell pair, we have that $\langle\Phi_2|P_0\otimes P_0|\Phi_2\rangle=\langle\Phi_2|P_1\otimes P_1|\Phi_2\rangle=1$
    
    Since $\widetilde{W}_1$ and $\widetilde{W}_2$ commute, their product satisfies $\widetilde W_1(x_1,y_1)\,\widetilde W_2(x_2,y_2)\ket{\Psi}=\ket{\Psi}$. Hence, the acceptance probability is 
    \begin{equation}
      \big\langle \Psi\big|\, \widetilde W_1(x_1,y_1)\,\widetilde W_2(x_2,y_2)\,\big|\Psi\big\rangle = \frac{1}{2}\cdot 1+ \frac{1}{2}\cdot 1=1, 
      \quad\text{for every question pair}.
    \end{equation}

    \textbf{Qubit count and compression:} In this example, each MSG needs $n_1=n_2=2$ qubits per player; the naive tensor baseline is $N=n_1+n_2=4$. This construction uses $n=2$ (data) $+\,1$ (control) $=3<4$, and satisfies the commuting-embeddings hypothesis of \Cref{theorem: all games}. \Cref{fig: two-msg-compress} illustrates this compression example. 
\end{example}

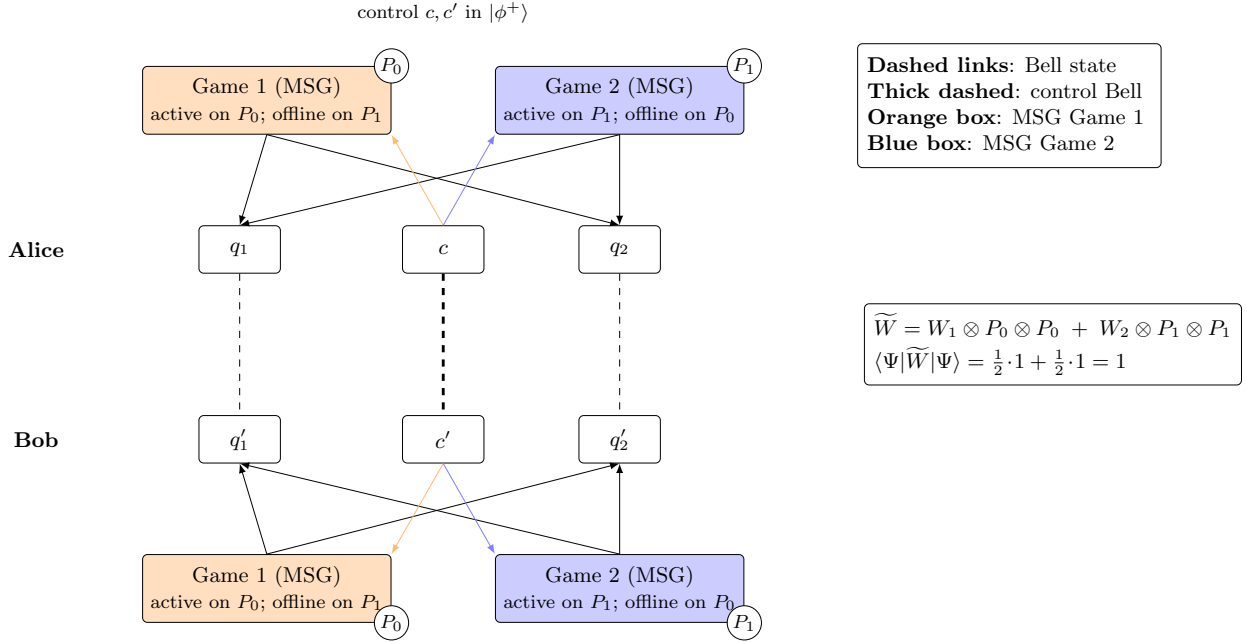
\begin{figure}[t]
\centering
\resizebox{\textwidth}{!}{%
\begin{tikzpicture}[font=\small, >=latex]

% Styles
\tikzset{
  qubit/.style={draw, rectangle, rounded corners=2pt, minimum width=1.2cm, minimum height=0.7cm},
  gbox1/.style={draw, rounded corners=2pt, fill=orange!25, minimum width=3.6cm, minimum height=1.0cm, align=center},
  gbox2/.style={draw, rounded corners=2pt, fill=blue!20,   minimum width=3.6cm, minimum height=1.0cm, align=center},
  legend/.style={draw, rounded corners=2pt, inner sep=4pt, align=left}
}

% --- Row labels
\node at (-1.8, 2.8) {\textbf{Alice}};
\node at (-1.8, 0.0) {\textbf{Bob}};

% --- Coordinates (wider spacing)
% Alice qubits: q1 at 0, q2 at 3.2, control at 6.8
\node[qubit] (Aq1) at (1.2,  2.8) {$q_1$};
\node[qubit] (Aq2) at (6.8,  2.8) {$q_2$};
\node[qubit] (Ac)  at (4.2,  2.8) {$c$};

% Bob qubits: q1' at 0, q2' at 3.2, control at 6.8
\node[qubit] (Bq1) at (1.2,  0.0) {$q'_1$};
\node[qubit] (Bq2) at (6.8,  0.0) {$q'_2$};
\node[qubit] (Bc)  at (4.2,  0.0) {$c'$};

% --- Bell states (dashed)
\draw[dashed] (Aq1.south) -- (Bq1.north);
\draw[dashed] (Aq2.south) -- (Bq2.north);
\draw[dashed, very thick] (Ac.south)  -- (Bc.north);

% --- Game boxes (Alice side, above)
\node[gbox1] (G1A) at (1.6, 5.0) {Game 1 (MSG)\\\footnotesize active on $P_0$; offline on $P_1$};
\node[gbox2] (G2A) at (6.8, 5.0) {Game 2 (MSG)\\\footnotesize active on $P_1$; offline on $P_0$};

% --- Game boxes (Bob side, below)
\node[gbox1] (G1B) at (1.6, -2.2) {Game 1 (MSG)\\\footnotesize active on $P_0$; offline on $P_1$};
\node[gbox2] (G2B) at (6.8, -2.2) {Game 2 (MSG)\\\footnotesize active on $P_1$; offline on $P_0$};

% --- Arrows from boxes to data qubits (Alice)
\draw[->] (G1A.south) -- (Aq1.north);
\draw[->] (G1A.south) -- (Aq2.north);
\draw[->] (G2A.south) -- (Aq1.north);
\draw[->] (G2A.south) -- (Aq2.north);

% --- Arrows from boxes to data qubits (Bob)
\draw[->] (G1B.north) -- (Bq1.south);
\draw[->] (G1B.north) -- (Bq2.south);
\draw[->] (G2B.north) -- (Bq1.south);
\draw[->] (G2B.north) -- (Bq2.south);

% --- Control selection annotations
\node at (4.2, 6.3) {\footnotesize control $c,c'$ in $\ket{\phi^+}$};
\draw[->, orange!55] (Ac.north) -- (G1A.south east);
\draw[->, blue!50] (Ac.north) -- (G2A.south west);
\draw[->, orange!55] (Bc.south) -- (G1B.north east);
\draw[->, blue!50] (Bc.south) -- (G2B.north west);

% P0/P1 tags on boxes
\node[draw, circle, inner sep=1pt, fill=white] at (G1A.north east) {\footnotesize $P_0$};
\node[draw, circle, inner sep=1pt, fill=white] at (G2A.north east) {\footnotesize $P_1$};
\node[draw, circle, inner sep=1pt, fill=white] at (G1B.south east) {\footnotesize $P_0$};
\node[draw, circle, inner sep=1pt, fill=white] at (G2B.south east) {\footnotesize $P_1$};

% --- Acceptance-operator summary (far right)
\node[legend, anchor=west] (Eq) at (10.4, 1.4) {%
\(\displaystyle \widetilde W = W_1\otimes P_0\otimes P_0\;+\; W_2\otimes P_1\otimes P_1\)\\[3pt]
\(\langle \Psi|\widetilde W|\Psi\rangle = \tfrac12\!\cdot\!1 + \tfrac12\!\cdot\!1 = 1\)
};

% --- Legend (top-left)
\node[legend, anchor=south west] at (10.3, 4.0) {%
\begin{tabular}{@{}l@{}}
\textbf{Dashed links}: Bell state\\
\textbf{Thick dashed}: control Bell\\
\textbf{Orange box}: MSG Game 1\\
\textbf{Blue box}: MSG Game 2
\end{tabular}
};

\end{tikzpicture}%
}
\caption{Two MSGs from~\Cref{ex: compressed} compressed to 3 qubits per player.}
\label{fig: two-msg-compress}
\end{figure}

\begin{remark}
    Dimension witness~\cite{brunner2013dimension} is a tool used to allow one to test the dimension of an unknown physical system in a device-independent manner, that is, without placing assumptions about the functioning of the devices used in the experiment. From \Cref{theorem: all games}, if all $K$ games can be embedded into commuting subalgebras on $\mathbb{C}^{2^n}$, for some $n<\sum_in_i$ qubits, then a dimension witness could, in principle, experimentally certify whether such a compressed embedding exists. i.e., use a dimension witness as a lower bound.
\end{remark}

    We summarize the key results into a pseudo-algorithmic version of how to leverage multiple NLGs, each of which has guaranteed perfect strategy (proven in \Cref{sec: dependent}), to invoke a parallel computation in order to generate a strategy for the game that utilizes fewer qubits (proven below in \Cref{theorem: qubit reduction}). The pseudocode is represented in Algorithm \ref{alg:glued}.

\SetKwComment{Comment}{/* }{ */}
\SetAlgoNlRelativeSize{0} % smaller line numbers
\SetNlSty{textbf}{(}{)}   % bold line numbers in parentheses
\SetAlgoNlRelativeSize{-1} % even smaller if desired
\SetKwInOut{Input}{Input}
\SetKwInOut{Output}{Output}

\begin{algorithm}[h]
    \caption{Game Compression via Commuting Embeddings}
    \label{alg:glued}
    \Input{Four NLGs $\{\mathcal{G}_i\}_{i=1}^4$, each with a perfect quantum strategy on $n_i$ local qubits.}
    \Output{A compressed game strategy on $n < N = \sum_i n_i$ local qubits.}
    
    %\Comment*[r]{Phase 1: Preparation and Dimension Determination}
    $N \gets \sum_i n_i$ \Comment*[l]{Naive additive qubit count}
    $n \gets \min\{m\in\mathbb{N}: 2^m-1\geq \max\{\mathrm{dim}(\mathfrak k_A)\}, \mathrm{dim}(\mathfrak k_B)\}$ \Comment*[l]{where $\mathfrak k_A,\mathfrak k_B\subset \mathfrak{su}(2^n)$}
    
    Identify $\mathcal{A}_i, \mathcal{B}_i$ \Comment*[l]{$C^*$-algebras from measurement operators $\{M_{x,a}\}_a$ and $\{N_{y,b}\}_b$}
    Identify $\mathfrak{g}_{\mathcal{A}_i}, \mathfrak{g}_{\mathcal{B}_i}$ \Comment*[l]{Corresponding Lie algebras} %generated by the measurement operators and taking nested commutators of the terms in that set.
    
    %\Comment*[r]{Phase 2: Constructing Commuting Embeddings}
    Construct embeddings 
    $\iota^A_i: \mathcal{A}_i \to \mathcal{B}(\mathbb{C}^{2^n})$, 
    $\iota^B_i: \mathcal{B}_i \to \mathcal{B}(\mathbb{C}^{2^n})$ \Comment*[l]{Injective $\ast$-homomorphisms}
    Enforce $[\iota^A_i(a), \iota^A_j(a')] = 0$, $[\iota^B_i(b), \iota^B_j(b')] = 0$ for $i \neq j$ \Comment*[l]{Pairwise commutativity}
    
    $\tilde{M}_{x_i,a_i} \gets \iota^A_i(M_{x_i,a_i})$ \Comment*[l]{Embedded POVM}
    $\tilde{N}_{y_i,b_i} \gets \iota^B_i(N_{y_i,b_i})$ \Comment*[l]{Embedded POVM}
    
    Choose entangled state $|\Psi\rangle \in \mathbb{C}^{2^{n}} \otimes \mathbb{C}^{2^{n}}$ \; 

    Confirm $n<N$\;
    
    \Return $\{\tilde{M}_{x_i,a_i}, \tilde{N}_{y_i,b_i}, |\Psi\rangle \}$

\end{algorithm}

    \section{Conclusion}\label{sec: conclusion}
In this work, we introduced an algebraic framework for playing multiple non-local games in parallel using fewer qubits than the naive tensor product baseline. The key observation is that perfect strategies can be consolidated whenever associated game algebras embed into commuting subalgebras. Two compression results were proved. First, for a referee who selects a single game at random from a finite family of games, we showed that one maximally entangled state of dimension equal to the largest game suffices. This allows the user to avoid repeating state preparation. Second, for parallel play, we proved that commuting embeddings of each game algebra enables simultaneous perfect strategies on few qubits. This reduction replaces additive dimensional growth. 

Following this line of work, several avenues remain. Our qubit bounds are certified by commuting embeddings and Cartan decomposition alignments. These results have yet to be extended to robust (noisy or approximate) scenarios, which would be ideal for near term devices. Furthermore, analyzing the computational complexity for the embeddings and the Cartan decompositions would complement our work well. Finally, incorporating constraints from limited connectivity in distributed and multi-QPU architectures to apply these results and integrating self-testing guarantees in this regime are all promising directions to follow this work.

With this work, we open a pathway to run richer device-independent tests and parallel quantum protocols on constrained hardware, with clear pathways to device verification, benchmarking, and implementation.

\section{Acknowledgments}\label{sec: acknowledgments}
The authors would like to thank Kathleen Hamilton and Eugene Dumitrescu for valuable discussions and insights that inspired this work.

This work was supported by Oak Ridge National Laboratory's (ORNL) Laboratory Directed Research and Development (LDRD) Seed Program. This work was partially supported by the U.S. Department of Energy, Office of Science, Office of Nuclear Physics Quantum Horizons: QIS Research and Innovation for Nuclear Science program at ORNL under FWP ERKBP91.

\end{document}